\documentclass[12pt,onecolumn,draft]{IEEEtran} 

\usepackage{amsmath,amssymb,amsfonts}
\usepackage[final]{graphicx}
\usepackage{psfrag}
\usepackage{subfigure}
\usepackage[numbers,sort&compress]{natbib}
\usepackage{mathrsfs} 
\usepackage{epstopdf}
\usepackage{multirow}
\usepackage{amsthm}
\usepackage{float}
\usepackage[nolist,nohyperlinks]{acronym}
\usepackage{svg}
\usepackage{xparse}
\NewDocumentCommand{\evalat}{sO{\big}mm}{%
  \IfBooleanTF{#1}
   {\mleft. #3 \mright|_{#4}}
   {#3#2|_{#4}}%
}

\begin{acronym}
\acro{PLS}{physical layer security}
\acro{SNR}{signal-to-noise ratio}
\acrodefplural{SNR}[SNRs]{signal-to-noise ratios}
\acro{AWGN}{additive white Gaussian noise}
\acro{CSI}{channel state information}
\acro{PDF}{probability density function}
\acrodefplural{PDF}[PDFs]{probability density functions}
\acro{CDF}{cumulative distribution function}
\acrodefplural{CDF}[CDFs]{cumulative distribution functions}
\acro{LRS}{large reflecting surface}
\acro{BS}{base station}
\acro{RV}{random variable}
\acro{FN}{folded normal}
\acro{TDD}{time-division duplexing}
\acro{LOS}{line-of-sight}
\acro{DL}{downlink}
\acro{MRT}{maximal ratio transmission}
\acro{MRC}{maximal ratio combining}
\acro{MC}{Monte Carlo}
\end{acronym}

\newtheorem{theorem}{Theorem}

\newtheorem{lemma}{Lemma}

\newtheorem{remark}{Remark}

\makeatletter
\def\blfootnote{\xdef\@thefnmark{}\@footnotetext}
\makeatother

\begin{document}
\title{{Physical Layer Security of Large Reflecting Surface Aided Communications with Phase Errors}}
\author{Jos\'e David Vega S\'anchez, Pablo Ram\'irez-Espinosa and F. Javier L\'opez-Mart\'inez}

\maketitle

\blfootnote{\noindent Manuscript received MONTH xx, YEAR; revised XXX. The review of this paper was coordinated by XXXX.  The work of J.~D.~Vega~S\'anchez was funded by the Escuela Polit\'ecnica Nacional, for the development of the project PIGR-19-06 and through a teaching assistant fellowship for doctoral studies. The work of F.J. Lopez-Martinez was funded by the Spanish Government and the European Fund for Regional Development FEDER (project TEC2017-87913-R) and by Junta de Andalucia (project P18-RT-3175, TETRA5G).}

\blfootnote{\noindent J.~D.~Vega~S\'anchez is with Departamento de Electr\'onica, Telecomunicaciones y Redes de Informaci\'on, Escuela Polit\'ecnica Nacional (EPN),
Quito,  170525, Ecuador. (e-mail: $\rm jose.vega01@epn.edu.ec$).}

\blfootnote{\noindent P. Ram\'irez-Espinosa is with the Connectivity Section,  Department of Electronic Systems, Aalborg University, Aalborg {\O}st 9220, Denmark (e-mail: $\rm pres@es.aau.dk$).}

\blfootnote{\noindent F.~J. Lopez-Martinez is with Departamento de Ingenieria de Comunicaciones, Universidad de Malaga - Campus de Excelencia Internacional Andalucia Tech., Malaga 29071, Spain (e-mail: $\rm fjlopezm@ic.uma.es$).}

\blfootnote{This work has been submitted to the IEEE for publication. Copyright may be transferred without notice, after which this version may no longer be accessible.} 
\begin{abstract}
The \ac{PLS} performance of a wireless communication link through a \ac{LRS} with phase errors is analyzed. Leveraging recent results that express the \ac{LRS}-based composite channel as an equivalent scalar fading channel, we show that the eavesdropper's link is Rayleigh distributed and independent of the legitimate link. The different scaling laws of the legitimate and eavesdroppers signal-to-noise ratios with the number of reflecting elements, and the reasonably good performance even in the case of coarse phase quantization, show the great potential of \ac{LRS}-aided communications to enhance \ac{PLS} in practical wireless set-ups.
\end{abstract}


\vspace{-2mm}
\section{Introduction}
Recently, large reflecting surfaces (LRSs) have been proposed as a new paradigm to noticeably improve the performance of emerging networks in terms of system performance and energy-efficiency. An LRS consists of a large number of low-cost passive reflecting units, where each element can adaptively adjust the amplitude reflection and/or the phase shift of the incident signals \cite{Wu}. These smart passive devices can be integrated into the infrastructure of future wireless networks to control the radio propagation environment.

On the other hand, physical layer security (PLS) has drawn full attention for ensuring secure wireless communications in a low complexity manner. Specifically, PLS intelligently exploits the inherent randomness of the wireless medium to protect the information in the physical layer \cite{diana}. From an information-theoretic perspective, LRS is a new approach to improve the PLS performance by reconfiguring the wireless environment for the benefit of the legitimate user. In this sense, several researchers have addressed their efforts to investigate PLS on LRS-aided wireless communications systems. For instance, the secrecy performance for LRS-aided multi-antenna communications was studied in \cite{Cui,Schober,Pei}. Because of the rather complex nature of the \ac{LRS} composite fading model, the analytical characterization of \ac{PLS} performance metrics is utterly unfeasible and most works often resort to optimization techniques to maximize the secrecy rates.


In this paper, we investigate the performance of an \ac{LRS}-aided communication system with imperfect phase compensation in terms of its PLS performance. We leverage the recent formulation of the \ac{LRS} composite fading channel as an equivalent scalar channel \cite{Badiu2020} to gain an understanding of the potential of \ac{LRS}-based communications for \ac{PLS}. The key contributions of this paper are: first, we show that the distribution of the eavesdropper's equivalent scalar fading channel is Rayleigh distributed and its average \ac{SNR} scales with $n$, while the average \ac{SNR} at the legitimate receiver scales with $n^2$. We also prove that despite the equivalent channels at both receivers share a number of components, they are statistically independent under some mild conditions. Finally, we exemplify the limitations of the equivalent scalar channel approximations for conventional asymptotic high-\ac{SNR} analyses, which should be interpreted with caution for outage-based performance metrics.
\vspace{-2mm}

\section{System Model}
We consider an LRS-assisted wireless communication set-up consisting of one source node Alice ($\mathrm{A}$), one legitimate node Bob ($\mathrm{B}$), one eavesdropper Eve ($\mathrm{E}$), and an LRS, which assists the communication between the legitimate nodes. In the system, the direct link is neglected, and all terminals are assumed to be equipped with a single antenna, while the LRS has $n$ low-cost passive reflecting elements $R_1 \dots R_n$. We denote as $H_{i,1}$ the fading channel coefficient between the source $\mathrm{A}$ and the reflecting element $R_i$, whereas $H_{i,{\rm b}}$ and $H_{i,{\rm e}}$ are the fading channel coefficients between $R_i$ and the legitimate receiver $\mathrm{B}$ and the eavesdropper $\mathrm{E}$, respectively. Without loss of generality, we consider normalized fading coefficients with unitary power, and the corresponding average magnitudes are given $\forall i=1\ldots n$ by $a_1=\mathbb{E}\{|H_{i,1}|\}$, $a_{2,\rm b}=\mathbb{E}\{|H_{i,{\rm b}}|\}$ and $a_{2,\rm e}=\mathbb{E}\{|H_{i,{\rm e}}|\}$. We note that $\{a_1,a_{2,\rm b},a_{2,\rm e}\}\leq1$ in all instances, 
where the equality only holds in the limit of a deterministic fading channel, i.e., in the absence of fading. For the sake of compactness, $a_{\rm b}=\sqrt{a_1 a_{2,\rm b}}$ and $a_{\rm e}=\sqrt{a_1 a_{2,\rm e}}$ are defined. The received signal at $\mathrm{B}$ can be expressed as 
\begin{equation}
\label{eqY}
Y_{\rm b}=\sqrt{P_T L_{\rm b}} \sum_{i=1}^{n} H_{i,1} e^{j \phi_{i}} H_{i,{\rm b}} X+W_b,
\end{equation}
where $X$ is the transmitted symbol, $P_T$ indicates the transmit power at A, $L_{\rm b}$ encompasses the path losses for the A-R and R-B links, the antenna gains and reflection losses, and $W_b$ is the \ac{AWGN} term with $N_0$ power. Now, the \ac{LRS} designs the phase shifts for each element $\phi_i$ so that all phase contributions due to $\angle H_{i,1}$ and $\angle H_{i,{\rm b}}$ are compensated. However, the imperfect phase estimation and the limited quantization of phase states at the \ac{LRS} causes that a residual random phase error $\Theta_i$ still persists \cite{Badiu2020}, i.e., $\phi_i=-\angle H_{i,1}-\angle H_{i,{\rm b}}+\Theta_i$. The equivalent complex channel observed by the legitimate receiver can hence be expressed as
\begin{equation}
\label{eq2}
H_{\rm b}=\frac{1}{n} \sum_{i=1}^{n}\left|H_{i,1} \| H_{i,{\rm b}}\right| e^{j \Theta_{i}},
\end{equation}
and \eqref{eqY} is reformulated as:
\begin{equation}
Y_{\rm b}=n \sqrt{P_T L_{\rm b}} H_b X+W_b
\end{equation}
Now, the received signal at $\mathrm{E}$ can be expressed as
\begin{equation}
Y_{\rm e}=\sqrt{P_T L_{\rm e}} \sum_{i=1}^{n} H_{i,1} e^{j \phi_{i}} H_{i,{\rm e}} X+W_e,
\end{equation}
where the $L_{\rm e}$ and $W_{\rm e}$ are defined in a similar way as $L_{\rm b}$ and $W_{\rm b}$. Because the phase shifts $\phi_{i}$ are designed to compensate for the effect of the fading channel coefficients of the legitimate link, the residual phase errors $\Psi_i$ affecting the eavesdropper link will be much larger than the legitimate counterpart and, whenever $\angle H_{i,{\rm e}}\sim\mathcal{U}[-\pi,\pi)$, then $\Psi_{i}\sim\mathcal{U}[-\pi,\pi)$ \cite{Scire1968} regardless of the phase distribution of $\angle H_{i,1}$. We can define the equivalent complex channel observed by the eavesdropper as
\begin{equation}
\label{eq5}
H_{\rm e}=\frac{1}{n} \sum_{i=1}^{n}\left|H_{i,1} \| H_{i,{\rm e}}\right| e^{j \Psi_{i}},
\end{equation}
that yields
\begin{equation}
Y_{\rm e}=n \sqrt{P_T L_{\rm e}} H_e X+W_e
\end{equation}
With the previous definitions, the instantaneous \ac{SNR} at the legitimate and eavesdropper's links are given by
\begin{equation}
\gamma_{\rm b}=n^2 \gamma_{0,{\rm b}} |H_{\rm b}|^2,
\end{equation}
\begin{equation}
\gamma_{\rm e}=n^2 \gamma_{0,e} |H_{\rm e}|^2,
\end{equation}
%
%
where we defined $\gamma_{0,{\rm b}}=P_T L_{\rm b}/N_0$
and $\gamma_{0,{\rm e}}=P_T L_{\rm e}/N_0$ as the average \ac{SNR}s at the legitimate and eavesdropper's sides in the case of a single reflector \ac{LRS} (i.e., $n=1$).

We aim to determine the system performance in terms of its secrecy capacity $C_\mathrm{S}$ defined as~\cite{Wyner}
\begin{align}\label{eqCS}
C_\mathrm{S}&=\!\text{max}\left \{C_{\rm b}-C_{\rm e},0  \right \},
\end{align}
where $C_{\rm b}=\log_2(1+\gamma_{\rm b})$ and $C_{\rm e}=\log_2(1+\gamma_{\rm e})$ are the capacities of the main and eavesdropper channels, respectively. We first consider a passive eavesdropper for which Alice does not have \ac{CSI} knowledge. Under this premise, Alice can only transmit at a constant secrecy rate $R_{\mathrm{S}}$ and security will be compromised whenever $R_{\mathrm{S}}$ exceeds $C_{\mathrm{S}}$. The secrecy outage probability (SOP) is formulated as the probability that the instantaneous $C_\mathrm{S}$ falls below such rate $R_{\mathrm{S}}$, i.e., $ \text{P}=\Pr\left \{ C_\mathrm{S}< R_{\mathrm{S}}  \right \}$ as
\begin{align}\label{eqSOP}
 \text{P}&=\int_{0}^{\infty}F_{\gamma_{\rm b} }\left ( \tau \gamma_{\rm e}+\tau-1 \right )f_{\gamma_{\rm e} }(\gamma_{\rm e} )d\gamma_{\rm e},
\end{align}
where $\tau\buildrel \Delta \over  = 2^{R_{\mathrm{S}}}$. We also study the active eavesdropping case, in which the CSI of both the main and the eavesdropper channels is available at Alice. Therefore, Alice can use such information to adapt her rate. In this setup, the average secrecy capacity (ASC) is the usual metric to evaluate the secrecy performance. According to~\cite[Proposition 3]{Moualeu}, the ASC can be defined as  
\begin{align}\label{eq22}
\overline{C}_\mathrm{S}=\overline{C}_{\rm b}-\mathcal{L}\left ( \overline{\gamma}_{\rm b}, \overline{\gamma}_{\rm e}\right ), 
\end{align}
where $\overline{C}_{\rm b}$ is the average capacity of the legitimate link 
 and $\mathcal{L}\left ( \overline{\gamma}_{\rm b}, \overline{\gamma}_{\rm e}\right )$ can be interpreted as an ASC loss, defined as
\begin{align}\label{eq24}
\mathcal{L}\left ( \overline{\gamma}_{\rm b}, \overline{\gamma}_{\rm e}\right )=\tfrac{1}{\ln 2}\int_{0}^{\infty}\tfrac{\left(1-F_{\gamma_{{\rm e} }}(\gamma_{\rm e} ) \right)\left( 1-F_{\gamma_{{\rm b} }}(\gamma_{\rm e} )\right)}{1+\gamma_{\rm e} }d\gamma_{\rm e}\geq 0.
\end{align}

\section{\ac{SNR} distributions}
\subsection{Distribution of $\gamma_{\rm b}$}
For sufficiently large $n$, \cite{Badiu2020} proved that the distribution of $H_{\rm b}$ is that of a non-circularly symmetric complex Gaussian \ac{RV} with $U_{\rm b}=\Re(H_{\rm b})$ and $V_{\rm b}=\Im(H_{\rm b})$, so that $U_{\rm b} \sim \mathcal{N}\left(\mu, \sigma_{U_{\rm b}}^{2}\right)$ and $V_{\rm b} \sim \mathcal{N}\left(0, \sigma_{V_{\rm b}}^{2}\right)$, where the parameters of $\mu =\varphi_{1} a_{\rm b}^{2}$, $\sigma_{U_{\rm b}}^{2} =\frac{1}{2 n}\left(1+\varphi_{2}-2 \varphi_{1}^{2} a_{\rm b}^{4}\right)$ and $\sigma_{V_{\rm b}}^{2} =\frac{1}{2 n}\left(1-\varphi_{2}\right)$, and $\varphi_{j}$ are the $j^{\rm th}$ circular moments of $\Theta_{i}$. This implies that $R_{\rm b}=|H_{\rm b}|$ follows the Beckmann distribution \cite{Beckmann1962} and hence, the average \ac{SNR} at the legitimate receiver $\gamma_{\rm b}$ follows a (squared) Beckmann distribution 
%
which is fully characterized by the following set of parameters $K=\mu^2/(\sigma_{U_{\rm b}}^{2} +\sigma_{V_{\rm b}}^{2} )$, $q=\sigma_{U_{\rm b}} /\sigma_{V_{\rm b}}$ and $\overline\gamma_{\rm b}=\mathbb{E}\{\gamma_{\rm b}\}$. 
 We note that the parameters $K$ and $q$ have a similar definition as those of the Rician and Hoyt \cite{Hoyt1947} distributions, respectively. In the scenario under consideration, we have that
\begin{align}
\label{eqK}
K&=n \frac{\varphi_1^2 a_{\rm b}^4}{1-\varphi_1^2 a_{\rm b}^4},\\
q&=\sqrt{\frac{1+\varphi_{2}-2 \varphi_{1}^{2} a_{\rm b}^{4}}{1-\varphi_{2}}},\\
\overline\gamma_{\rm b}&=n^2 \gamma_{0,{\rm b}}\left[\varphi_1^2 a_{\rm b}^4+\frac{1}{n}\left(1-\varphi_1^2 a_{\rm b}^4\right)\right].
\end{align}

As stated in \cite{Badiu2020}, the average \ac{SNR} scales with $n^2$. We also observe that the \ac{LOS} condition of the equivalent scalar channel grows, captured by $K$, grows with $n$. Notably, the non-circular symmetry caused by the phase errors captured by $q\in[1,\infty)$ is independent of the number of elements of the \ac{LRS}. We note that in the absence of phase errors, then $H_{\rm b}$ becomes a real Gaussian \ac{RV} and hence $|H_{\rm b}|$ follows a folded normal (FN) distribution \cite{Reig2019} with parameter $K$ given by \eqref{eqK} with $\varphi_1=1$, and for which the PDF and CDF have a simple closed-form expression.

The distribution of $R_{\rm b}$ is well approximated by a Nakagami-$m$ distribution in \cite{Badiu2020}, and hence $\gamma_{\rm b}$ can be approximated by a gamma distribution with shape parameter $m=\frac{n}{2} \frac{\varphi_{1}^{2} a_{\rm b}^{4}}{1+\varphi_{2}-2 \varphi_{1}^{2} a_{\rm b}^{4}}$ and scale parameter $\overline\gamma_{\rm b}=n^2 \gamma_{0,{\rm b}}\varphi_1^2 a_{\rm b}^4$. Similarly to $K$, $m$ also scales with $n$, which is in coherence with the conventional approximation of the Rician distribution by a Nakagami-$m$ distribution \cite{Simon2005} -- only that in our case, we are approximating a generalization of the Rician distribution by a Nakagami-$m$ distribution. Because of the rather dissimilar behavior of the FN, the Beckmann and the Nakagami-$m$ distributions in terms of diversity order \cite{Wang2003}, we will consider all such distributions in the derivation of the \ac{PLS} performance metrics, in order to obtain insights on when these distributions are useful to approximate the true distribution of $\gamma_{\rm b}$.
\vspace{-3mm}
\subsection{Distribution of $\gamma_{\rm e}$}
When the \ac{LRS} designs its phase shifts according to the legitimate link, the resulting phase distributions for each of the eavesdropper's R-E links $\Psi_i$ are uniformly distributed by virtue of \cite{Scire1968}. This implies that the distribution of $R_{\rm e}=|H_{\rm e}|$ is Rayleigh distributed according to \cite[Corol. 2]{Badiu2020} with variance $\mathbb{E}\{R_{\rm e}^2\}=1/n$. Hence, $\gamma_{\rm e}$ is exponentially distributed with $\overline\gamma_{\rm e}=n\gamma_{0,{\rm e}}$. 

\begin{remark}[Scaling law for $\overline\gamma_{\rm e}$]
Notably, the average \ac{SNR} at the eavesdropper scales with $n$, whereas the average \ac{SNR} at the legitimate receiver scales with $n^2$. Hence, the scaling law for the ratio of legitimate and wiretap \ac{SNR}s is
\begin{equation}
\left.\frac{\overline\gamma_{\rm b}}{\overline\gamma_{\rm e}}\right|_{n\uparrow} 
= n \frac{\gamma_{0,{\rm b}}}{\gamma_{0,\rm e}}\left[\varphi_1^2 a_{\rm b}^4+\frac{1}{n}\left(1-\varphi_1^2 a_{\rm b}^4\right)\right]
\end{equation}
This implies that, as long as the operational assumptions for the \ac{LRS} hold, the use of a larger \ac{LRS} can provide an \ac{SNR} boost to the legitimate link compared to the eavesdropper's counterpart.
\end{remark}

Inspection of \eqref{eq2} and \eqref{eq5} reveals that the legitimate and eavesdropper's links share a common part through $H_{i,1}$. However, we will now prove that both equivalent channels are statistically independent.
\begin{theorem}[Independence of legitimate and wiretap links]\label{th1} Let us consider the equivalent legitimate and wiretap channels in \eqref{eq2} and \eqref{eq5}. Then, $H_b$ and $H_e$ are independent if $\angle H_{i,{\rm e}}\sim\mathcal{U}[-\pi,\pi)$. This is the case, e.g., of considering Rayleigh fading for the \ac{LRS} to eavesdropper's links.
\end{theorem}

\begin{proof}
See Appendix \ref{appendix1}.
\end{proof}

\section{\ac{PLS} Performance}
\label{PLSp}
In this section, we will derive analytical expressions for the chief \ac{PLS} performance metrics defined previously. We will consider three different scenarios for our analysis, which imply different approximations for the legitimate/wiretap links, respectively: (\emph{a}) no phase errors -- FN/Rayleigh case; (\emph{b}) phase errors -- Beckmann/Rayleigh case; (\emph{c}) phase errors -- Nakagami/Rayleigh case. For the sake of shorthand notation, we will refer to these scenarios with the subindices FR, BR and NR, respectively.
\vspace{-2mm}
\subsection{\text{SOP} Analysis}
\begin{lemma}[SOP in FR scenario]\label{Lema1}
 The SOP and the asymptotic SOP expressions ($\overline\gamma_{\rm b}\rightarrow\infty$) in the absence of phase errors for LRS-aided communications are given by 
 \end{lemma}
 \vspace{-7mm}
  \begin{align}\label{SopFR}
 \text{P}_{\text{FR}}&= 1-Q_{0.5}\left(a_0,b_0\right)+e^{\frac{\tau-1}{\tau\overline\gamma_{\rm e}}+c_s}\tfrac{a_s}{\sqrt{K}}Q_{0.5}\left(a_s,b_s\right),
\end{align}
 \begin{align}\label{AsymSopFR}
 \text{P}_{\text{FR}}^{\infty}\simeq e^{-K/2+\tfrac{\tau-1}{\tau\overline\gamma_{\rm e}}}\sqrt{\tfrac{\tau\overline\gamma_{\rm e}(1+K)}{2\overline\gamma_{\rm b}}}\tilde\Gamma\left(1.5,\tfrac{\tau-1}{\tau\overline\gamma_{\rm e}}\right),
\end{align}
with $\tau=2^{R_{\rm S}}$, $\tilde\Gamma(\cdot,\cdot)$ is the regularized upper incomplete Gamma function, $a_s=\sqrt{\frac{K(K+1)}{K+1-2\overline\gamma_{\rm b}s}}$, $b_s=\sqrt{2\left(\frac{K+1}{2\overline\gamma_{\rm b}}-s\right)z}$, $s=-\frac{1}{\tau\overline\gamma_{\rm e}}$, $z=\tau-1$ and $c_s=\frac{K\overline\gamma_{\rm b}s}{K+1-2\overline\gamma_{\rm b}s}$. The Marcum $Q$-function of order $0.5$ can be easily computed with the help of the Gaussian $Q$ function as $Q_{0.5}(a,b)=Q(b-a)+Q(b+a)$.
\begin{proof}
First, \eqref{SopFR} is obtained from \cite{Lopezz} by specializing the parameter of the $\kappa$-$\mu$ distribution to $\mu_{\kappa\text{-}\mu}=0.5$ and some manipulations. Then, \eqref{AsymSopFR} is obtained by using the approach in~\cite{Perim} with $\mu_{\kappa\text{-}\mu}=0.5$, and then substituting the resulting expression in \eqref{eqSOP} followed by some manipulations.
\end{proof}

\begin{lemma}[SOP in BR scenario]\label{Lema1b}
 The SOP and the asymptotic SOP expressions ($\overline\gamma_{\rm b}\rightarrow\infty$) considering phase errors in the LRS-aided communications are given by 
 \end{lemma}
 \vspace{-7mm}
  \begin{align}\label{SopBR}
 \text{P}_{\text{BR}}&= F_{\gamma_{\rm b}}(\tau-1) +\exp\left ( \tfrac{\tau-1}{\tau \overline\gamma_{\rm e}} \right ) \mathcal{M}_{\gamma_{\rm b}}^{u}\left ( -\tfrac{1}{\tau \overline\gamma_{\rm e}} ,\tau-1\right ).
\end{align}
 \begin{align}\label{AsymSopBR}
 \text{P}_{\text{BR}}^{\infty}\simeq & \exp\left ( - \tfrac{K\left ( 1+q^2 \right )}{2q^2} \right )\tfrac{\left ( 1+K \right )\left ( 1+q^2 \right )\left ( \overline\gamma_{\rm e}\tau+\tau-1 \right ) }{2q\overline\gamma_{\rm b}}
\end{align}
where $F_{\gamma_{\rm b}}(\cdot)$~\cite[Eq.~(7)]{Lopezz} is the CDF of a squared Beckmann distribution, and $\mathcal{M}_{\gamma_{\rm b}}^{u}\left (\cdot,\cdot\right )$~\cite[Eq.~(3)]{Lopezz} is the upper-incomplete moment generating function (MGF) of the RV $\gamma_{\rm b}$, which follows a squared Beckmann distribution. The evaluation of $\mathcal{M}_{\gamma_{\rm b}}^{u}\left (\cdot,\cdot\right )$ is carried out numerically through an inverse Laplace transformation \cite{Abate1995} over a shifted and scaled version of the (conventional) MGF of $\gamma_{\rm b}$, as in \cite[Eq. 4]{Lopezz}, which is obtained from~\cite[Eq.~(2.41)]{Simon2005} with $r\rightarrow \infty$. 
\begin{proof}
$\text{P}_{\text{BR}}$ can be obtained directly from~\cite[Eq.~(21)]{Lopezz} with the respective substitutions. On the other hand, the $ \text{P}_{\text{BR}}^{\infty}$ is derived by using~\cite[Proposition 3]{Wang2003}, in which $d=1$, using the MGF of $\gamma_b$.
\end{proof}

\begin{lemma}[SOP in NR scenario]\label{Lema2}
 The SOP and the asymptotic SOP expressions ($\overline\gamma_{\rm b}\rightarrow\infty$) considering phase errors in the LRS-aided communications can be approximated as 
 \end{lemma}
 \vspace{-7mm}
 \begin{align}\label{SopNR}
 \text{P}_{\text{NR}}&={\tilde\gamma\left ( m,\tfrac{(\tau-1) m}{\overline\gamma_{\rm b}} \right )}+e^{\tfrac{\tau-1}{\tau\overline\gamma_{\rm e}}} \tfrac{\tilde \Gamma\left (m,\left ( \tau-1 \right )\left ( \tfrac{m}{\overline\gamma_{\rm b}} +\frac{1}{\tau \overline\gamma_{\rm e}}\right ) \right ) }{\left (1+\tfrac{\overline\gamma_{\rm b}}{m\tau \overline\gamma_{\rm e}}\right )^m },\\
 \label{AsymSopNR}
 \text{P}_{\text{NR}}^{\infty}&\simeq e^{\tfrac{\tau-1}{\tau\overline\gamma_{\rm e}}}\left ( \tfrac{\tau m \overline\gamma_{\rm e}}{\overline\gamma_{\rm b}} \right )^m  \tilde\Gamma\left ( m+1,\tfrac{\tau-1}{\tau\overline\gamma_{\rm e}} \right ) .
\end{align}
where $\tilde\gamma(\cdot,\cdot)$ is the regularized lower incomplete Gamma functions~\cite[Eq.~(8.350.1)]{Gradshteyn}.
\begin{proof}
As in the proof of Lemma 2, \eqref{SopNR} is obtained from \cite{Lopezz} by setting the parameters of the $\kappa$-$\mu$ distribution $\kappa=0$ and $\mu=m$. Then, \eqref{AsymSopNR} is obtained as a particular case of~\cite[Eq.~(21)]{Moualeu} with the respective substitutions.
\end{proof}

Inspection of \eqref{AsymSopFR}, \eqref{AsymSopBR} and \eqref{AsymSopNR} reveals a different secrecy diversity order for each of the approximations, i.e., $1/2$, $1$ and $m$ for the FR, BR and NR cases, respectively. The implications arising from this observation will be discussed in the Numerical Results section.

\subsection{ASC Analysis}
For the sake of compactness, we will use a common formulation for the ASC metrics in the FR, NR and BR scenarios.
\begin{lemma}\label{LemaASC}
 The ASC and the asymptotic ASC ($\overline\gamma_{\rm b}\rightarrow\infty$) formulations over Z/Rayleigh fading channels for LRS-aided  communications can be obtained as
\begin{align}\label{ExactCs}
 \overline{C}_\mathrm{S} = 
  & \overline{C}_{\rm B} -  \overline{C}_{\rm E} + \mathcal{G}_{\rm Z}\left(\overline\gamma_{\rm b},\overline\gamma_{\rm e}\right)
\end{align}
\vspace{-7mm}
\begin{align}\label{AsymCs}
 \overline{C}_\mathrm{S}^{\infty}&\approx  \overline{C}_{\rm B} -  \overline{C}_{\rm E},\\
\label{AsymCs2}
  &\approx \log_2\left (\overline\gamma_{\rm b}\right)-t_{\rm Z} -\overline{C}_\mathrm{E},
\end{align}
where Z=$\{\text{Folded\,Normal},\text{Beckmann},\text{Nakagami}\}$ indicates the distribution of the legitimate link, and $t_{\rm Z}$ is a constant value that captures the fading severity loss of the legitimate link \cite{Moualeu}. We note that $\overline{C}_\mathrm{E}=\tfrac{e^{1/\overline\gamma_{\rm e}} }{\ln 2}E_1\left ( \tfrac{1}{\overline\gamma_{\rm e}} \right )$ denotes the average capacity of the wiretap link under the Rayleigh approximation, with $E_1\left (\cdot \right)$ being the Exponential integral function, and the term $\mathcal{G}_{\rm Z}\left(\overline\gamma_{\rm b},\overline\gamma_{\rm e}\right)=\tfrac{e^{1/\overline\gamma_{\rm e}} }{\ln 2}\int_{0}^1 \tfrac{1}{u}e^{-1/(u\overline\gamma_{\rm e})}M_{\gamma_{\rm b}}\left(\tfrac{-1}{u\overline\gamma_{\rm e}}\right)du\geq0$, where $\mathcal{M}_{\gamma_{\rm b}}\left (\cdot\right )$ is the (conventional) MGF of $\gamma_{\rm b}$.
 \end{lemma}
 
 \begin{proof}
See Appendix \ref{appendix2}.
\end{proof}
The previous Lemma allows us to evaluate the ASC in the investigated scenario in a compact form. $\overline{C}_{\rm B}$ can be either computed using available results in the literature \cite{Costa2007,Moreno2016}, or evaluated through numerical integration or quadrature methods. We note that as pointed out in \cite{Moualeu}, the term $\mathcal{G}_{\rm Z}\left(\overline\gamma_{\rm b},\overline\gamma_{\rm e}\right)$ vanishes as $\overline\gamma_{\rm b}$ grows, which in our case happens as $n$ is increased.

\vspace{-1mm}
\section{Numerical evaluation}

We now evaluate the effect of phase errors on the secrecy performance metrics in the investigated scenario, as well as the goodness of the scalar approximations for the equivalent composite channel in \ac{LRS}-assisted communications. For the links between $\mathrm{A}$ and the \ac{LRS}, and between the \ac{LRS} and $\mathrm{B}$, we consider Rician fading with parameter $K=1$. The links between the \ac{LRS} and $\mathrm{E}$ are assumed to be Rayleigh distributed. Hence, we have that $a_1= a_{2,\rm b}=\sqrt{\pi/(4(K+1))}{}_1F_1\left(-1/2,1,-K\right)$, where ${}_1F_1(\cdot)$ is Kummer hypergeometric's function, and $a_{2,\rm e}=\sqrt{\pi}/2$. For the sake of brevity, we consider phase errors due to the finite number of phase shifts available at the \ac{LRS}, although similar conclusions can be extracted by considering the phase estimation error model \cite{Badiu2020}; hence, the phase errors are uniformly distributed in the interval $[-u_{n_{b}},u_{n_{b}}]$ with $u_{n_{\rm b}}=-2^{-n_{b}}\pi$, where $n_{\rm b}$ is the number of quantization bits used to encode the phase shifts. Thus, from \cite{Badiu2020} we have $\varphi_i=\tfrac{\sin\left(u_{n_{\rm b}+1-i}\right)}{u_{n_{\rm b}+1-i}}$ for $i=\{1,2\}$.

In the next figures, we set $\overline\gamma_{\rm0,e}=10$ dB, a fixed transmit power $P_T$, and study the effect of increasing $n_{\rm b}$; the ideal case of no phase errors is included as a reference in all instances. The exact values for the secrecy metrics are obtained through Monte Carlo (MC) simulations. The analytical secrecy performance metrics in the FR, NR, and BR cases are included using the results in Section \ref{PLSp}. These have also been double-checked offline with additional MC simulations, which are not included in the figures for the sake of readability.

Fig. \ref{figASC1} shows the ASC as a function of $\overline\gamma_{\rm0,b}$, for different values of $n_{\rm b}$ and number of elements at the \ac{LRS} through $n$. Theoretical values have been evaluated with \eqref{ExactCs} and are represented using solid lines. Asymptotic values are computed with \eqref{AsymCs} for the BR case, and with \eqref{AsymCs2} for the FR and NR cases with $t_{\rm FR}$ and $t_{\rm NR}$ in \cite[Table II]{Moreno2016}. We extract important insights from the observation of Fig. \ref{figASC1}: (\emph{i}) increasing $n$ allows for improving the ASC for a fixed $\overline\gamma_{\rm0,b}$, thanks to the different scaling laws of the legitimate and wiretap average \ac{SNR}s; (\emph{ii}) FR (no phase errors) and BR (phase errors) equivalent scalar approximations work pretty well regardless of $n$, while the NR one underestimates the true ASC for low $n$; (\emph{iii}) asymptotic ASC expressions are tight for a wide range of SNR values; and (\emph{iv}) the performance degradation with $n_{\rm b}=2$ bits is small, which confirms that state-of-the-art solutions for \ac{LRS} surfaces \cite{Dai2020} may be enough to obtain a secrecy performance close to the case of no phase errors. Indeed, all previous remarks hold as long as the operating assumptions of the \ac{LRS} in terms of size as $n$ grows are valid.

\begin{figure}[t]
\psfrag{Z}[Bc][Bc][0.6]{$\mathrm{\textit{n}=256}$}
\psfrag{q1}[Bc][Bc][0.6]{$n_{\rm b}=1$}
\psfrag{q2}[Bc][Bc][0.6]{$n_{\rm b}=2$}
 \psfrag{Z}[Bc][Bc][0.6]{$\mathrm{\textit{n}=256}$}
\psfrag{Z1}[Bc][Bc][0.6]{$\mathrm{\textit{n}=64}$}
\psfrag{Z2}[Bc][Bc][0.6]{$\mathrm{\textit{n}=16}$}
\psfrag{Z3}[Bc][Bc][0.6]{$\mathrm{\textit{n}=4}$}           \includegraphics[width=\linewidth]{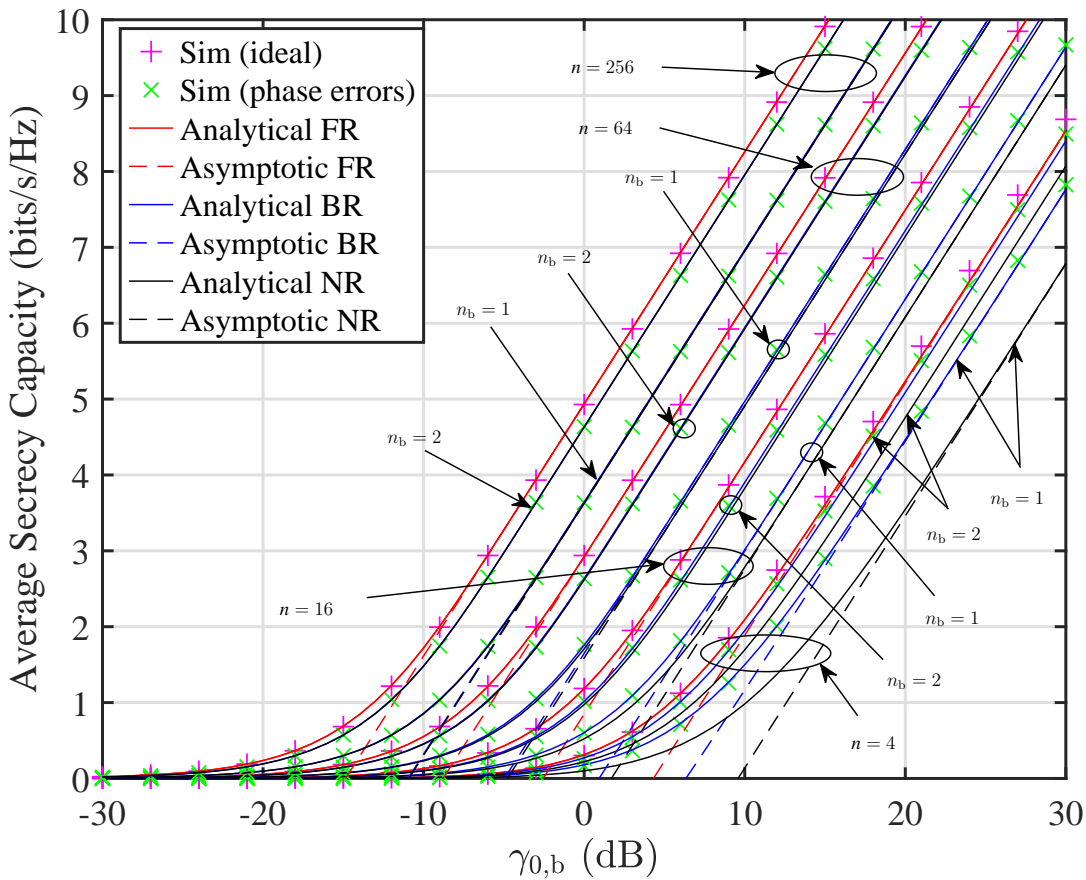}
\caption{ASC as a function of $\overline\gamma_{\rm0,b}$ for different values of $n_{\rm b}$ and $n$. Markers correspond to the legitimate and eavesdropper channels in \eqref{eq2} and \eqref{eq5}. }
\label{figASC1}
\vspace{-5mm}
\end{figure}
   \begin{figure}[t]
  \psfrag{Z}[Bc][Bc][0.6]{$\mathrm{\textit{n}=256}$}
\psfrag{Z1}[Bc][Bc][0.6]{$\mathrm{\textit{n}=64}$}
\psfrag{Z2}[Bc][Bc][0.6]{$\mathrm{\textit{n}=16}$}
\psfrag{Z3}[Bc][Bc][0.6]{$\mathrm{\textit{n}=4}$}       \includegraphics[width=\linewidth]{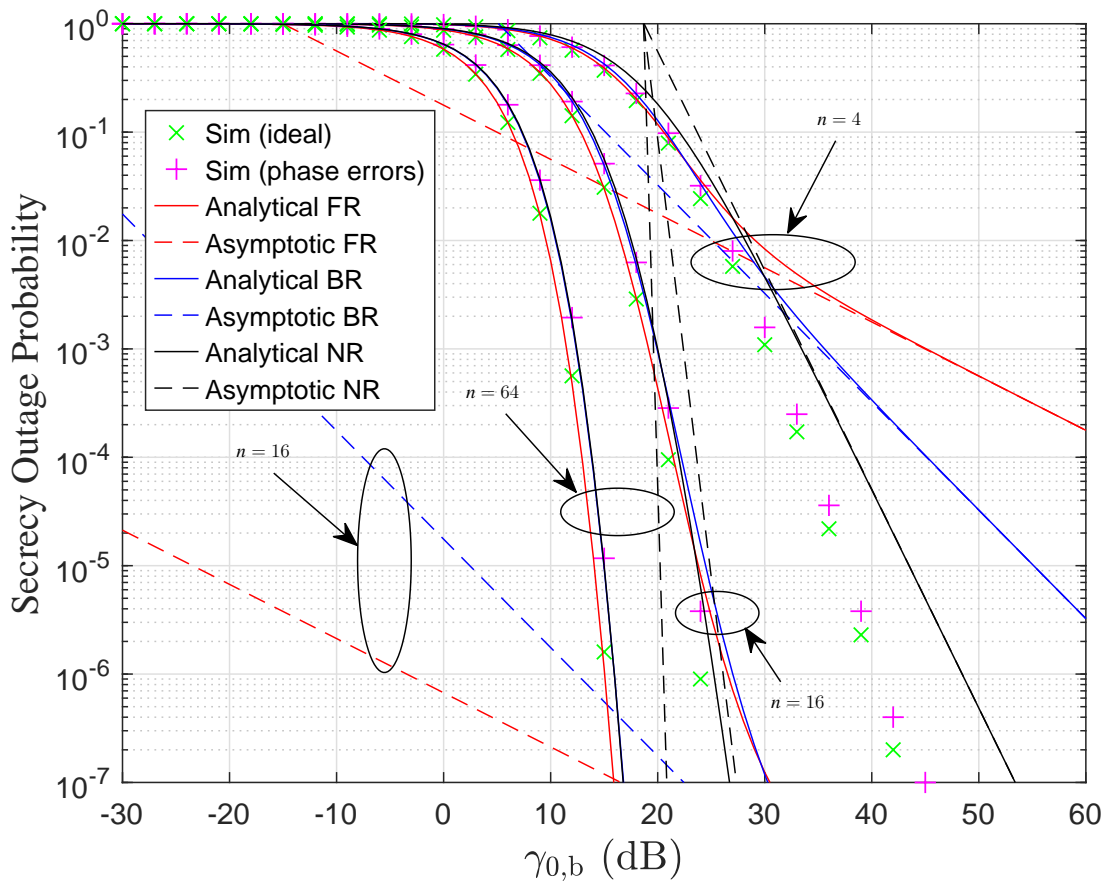}
     \caption{SOP as a function of $\overline\gamma_{\rm0,b}$ for different values of $n$ with $n_{\rm b}=2$ and $n_{\rm b}\rightarrow\infty$. Markers correspond to the legitimate and eavesdropper channels in \eqref{eq2} and \eqref{eq5}.}
\label{figOP1}
\vspace{-5mm}
 \end{figure}
          
Fig. \ref{figOP1} shows the SOP as a function of $\overline\gamma_{\rm0,b}$, for different values of $n$ and $n_{\rm b}=\{2,\infty\}$. Theoretical values have been evaluated with the expressions included in Lemmas 1 to 3. Similar conclusions as in the ASC can be extracted, especially confirming that $n_{\rm b}=2$ bits allow for a good performance compared to the ideal case. However, some relevant differences are observed: while the equivalent scalar approximations work well in all instances for large $n$, there are substantial differences between the exact simulated results and the FR, BR, and NR cases for lower $n$. More importantly, the asymptotic results for the SOP may induce to confusion if not interpreted properly: while all asymptotic results are tight (i.e., they all coincide with the analytical SOP expressions for each case), the different secrecy diversity order inherent to each of the equivalent scalar approximations is translated into a different decay of the high-SNR slopes. Because of the high line-of-sight condition of the FR and BR scalar approximations, the asymptotes kick-in at very low SOP values; conversely, the NR asymptote seems to better capture the abrupt decay of the SOP for the operating range of probability values. In any case, asymptotic analyses for the SOP should be exercised with caution when using the equivalent scalar approximations, as they may not be representative of the actual behavior of the real \ac{LRS}-assisted channel.       

\section{Conclusions}
The potential of \ac{LRS} for \ac{PLS} and the usefulness of equivalent scalar channel approximations for performance evaluation in such contexts have been exemplified, both theoretically and by simulation. Even when the \ac{LRS} has a limited phase resolution of $2$ bits, the different scaling laws for the desired and eavesdropper's SNRs allows for improving the \ac{PLS} performance in \ac{LRS}-assisted communications. The implications of using multiple antenna devices by all agents, and the potential impact of spatial correlation in the fading links are key aspects to be further investigated.


\appendices
\vspace{-3mm}
\section{Proof of Theorem \ref{th1}}
\label{appendix1}
Using the law of total expectation, and conditioning on the set $Z=\{H_{i,1}, H_{i,b},\Theta_{i}\}$, we can write $\mathbb{E}\{H_{\rm b}H_{\rm e}\}=\mathbb{E}\{\mathbb{E}\left\{H_{\rm b}H_{\rm e}|Z\right\}\}=\mathbb{E}\{H_{\rm b}\mathbb{E}\left\{H_{\rm e}|Z\right\}\}$. Now, the inner expectation can be expanded as
\begin{align}
\mathbb{E}\left\{H_{\rm e}|Z\right\}&=\frac{1}{n}\sum_{i=1}^{n}|H_{i,1}|\mathbb{E}\left\{|H_{i,\rm{e}}|e^{j \Psi_{i}}\right\}.
\end{align}
Now, by virtue of \cite{Scire1968} the distribution of $\Psi_{i}$ is uniform in any interval of length $2\pi$ provided that $\angle H_{i,\rm{e}}$ is uniformly distributed in the same interval. Under the mild assumption that $|H_{i,\rm{e}}|$ and $e^{j \Psi_{i}}$ are independent, which is the case for instance of $|H_{i,\rm{e}}|$ being Rayleigh distributed, then it yields that $\mathbb{E}\left\{H_{\rm e}|Z\right\}=0$. Hence, the independence between $H_{\rm b}$ and $H_{\rm e}$ is stated. This completes the proof.

\section{Proof of Theorem \ref{th1}}
\label{appendix2}
From the definitions in \eqref{eq22} and \eqref{eq24}, we use the expression of the CDF of the exponential distribution for the wiretap link. After integration by parts, two terms are identified; the first one corresponds to $\overline{C}_\mathrm{E}$ in \eqref{ExactCs}, whereas the second one reduces to $\mathcal{G}_{\rm Z}\left(\overline\gamma_{\rm b},\overline\gamma_{\rm e}\right)$ after: (\emph{i}) leveraging the integral definition of the Exponential integral function in \cite{Gradshteyn} in \eqref{ExactCs}, (\emph{ii}) changing the order of integration, and (\emph{iii}) using the definition of the MGF. As for the asymptotic ASC results, they hold by virtue of \cite[eq. (43)]{Moualeu}.

\vspace{-2mm}

\bibliographystyle{ieeetr}
\bibliography{bibfile}

\end{document}